\documentclass[reqno,11pt]{amsart}
\usepackage{fullpage}
\usepackage{amsfonts}
\usepackage{graphicx}
\usepackage{subfigure}
\usepackage{amssymb}
\usepackage{latexsym}
\usepackage{cite}
\usepackage{bm}
\usepackage[all]{xy}
\usepackage{arydshln}
\usepackage{amsmath,amsthm,accents}
\numberwithin{equation}{section}
\numberwithin{equation}{section}

\newtheorem{Thm}{Theorem}

\def \wh#1{\widehat{#1}}
\def \wt#1{\widetilde{#1}}

\newcommand{\nn}{\nonumber}
\newcommand{\bC}{\boldsymbol{C}}
\newcommand{\bD}{\boldsymbol{D}}
\newcommand{\bI}{\boldsymbol{I}}
\newcommand{\bK}{\boldsymbol{K}}
\newcommand{\bL}{\boldsymbol{L}}
\newcommand{\bM}{\boldsymbol{M}}
\newcommand{\br}{\boldsymbol{r}}
\newcommand{\bs}{\boldsymbol{s}}
\newcommand{\tc}{\,^t \hskip -2pt {\boldsymbol{c}}}
\newcommand{\bT}{\boldsymbol{T}}
\newcommand{\bF}{\boldsymbol{F}}
\newcommand{\bG}{\boldsymbol{G}}
\newcommand{\bH}{\boldsymbol{H}}

\newcommand{\Ga}{\boldsymbol{\Gamma}}

\newcommand{\Og}{\boldsymbol{\Omega}}
\newcommand{\st}{\hbox{\tiny\it{T}}}

\def \cd#1{\accentset{\bullet}{#1}}
\def \dc#1{\accentset{\circ}{#1}}
\DeclareMathOperator{\sech}{sech}
\DeclareMathOperator{\csch}{csch}

\begin{document}

\title{Local and nonlocal complex discrete and semi-discrete sine-Gordon equations and solutions}

\author{Xiao-bo Xiang, Wei Feng, Song-lin Zhao$^*$\\
\\\lowercase{\scshape{
Department of Applied Mathematics, Zhejiang University of Technology,
Hangzhou 310023, P.R. China}}}
\email{*Corresponding Author: songlinzhao@zjut.edu.cn}
\begin{abstract}

In this paper, local and nonlocal complex reduction of a discrete and a semi-discrete negative order Ablowitz-Kaup-Newell-
Segur equations is studied. Cauchy matrix type solutions,
including soliton solutions and Jordan-block solutions, for the resulting local and nonlocal complex discrete and semi-discrete sine-Gordon equations
are constructed. Dynamics of 1-soliton solution are analyzed and illustrated.

\end{abstract}


\maketitle

\section{Introduction}
\label{sec-1}

The nonlocal integrable systems have been quite widely studied in recent years.
The first example of nonlocal integrable systems is the reverse-space
nonlinear Schr\"{o}dinger equation
\begin{align}
\label{n-NLS}
iw_t(x,t)+w_{xx}(x,t)\pm w^2(x,t)w^*(-x,t)=0,
\end{align}
proposed by Ablowitz and Musslimani in \cite{AM-nNLS}, where asterisk means the complex conjugate.
This equation is parity-time symmetric because the potential $V(x,t)=w(x,t)w^*(-x,t)=V^*(-x,t)$. Since
then, the nonlocal integrable systems have attracted much interest in view of their wide range of
applications in various fields of mathematics and physics.
In mathematics, this type of equation possesses Lax integrability and admits infinite number of
conservation laws \cite{AM-nNLS-IST}. Physically speaking, this type of equation can be usually
applied to describing the two-place (Alice-Bob) physics or multi-place physics \cite{Lou-JMP,Lou-CTP}.
Besides the model \eqref{n-NLS}, many other nonlocal integrable systems have been proposed,
including nonlocal Korteweg-de Vries model, nonlocal modified Korteweg-de Vries model,
nonlocal sine-Gordon (sG) model, nonlocal Davey-Stewartson model, (e.g. \cite{Lou,AbMu-SAPM}).
Up to now, many traditional methods have been used to search for exact solutions to the nonlocal integrable systems
\cite{LiX-PRE-2015,HL,Yan-AML-2015,Zhou-ZX-1,XU-AML-2016,SXZ,Ab-IST,Yang-Chen,CZ,CDLZ,FZS-IJMPB,Zhang-SAPM}, such as the inverse scattering transformation,
the Riemann-Hilbert approach, the Hirota's bilinear method, the Darboux transformation, etc.

Although great progress has been got in the nonlocal systems within the field of the continuous or semi-discrete
integrable systems. There was a little work on the nonlocal discrete integrable systems.
In \cite{ZKZ}, Zhang {\it et al.} investigated two types of nonlocal discrete integrable equations, which were
constructed as reductions of 2-component Adler-Bobenko-Suris systems \cite{Bri}. In particular, they showed that
the $2[0,1]$ Adler-Bobenko-Suris system allows reverse-$n$ nonlocal reduction and
the $2[1,1]$ Adler-Bobenko-Suris system admits reverse-$(n,m)$ nonlocal reduction.
Moreover, they used the nonlocal H1 equations as examples to demonstrate their idea. Up to now, there was no
open literature on the solutions to the nonlocal discrete integrable systems to the best of the author's knowledge.

In this paper, a discrete Ablowitz-Kaup-Newell-Segur (AKNS) type equation that we want to investigate is of form
\begin{subequations}
\label{dnAKNS}
\begin{align}
& \label{dnAKNS-a} q(\wh{u}-u)[p^2-(u+\wt{u})(v+\wt{v})]^{\frac{1}{2}}+(u+\wt{u})[1-q^2(\wh{u}-u)(\wh{v}-v)]^{\frac{1}{2}}\nn \\
& \quad +(\wh{u}+\wh{\wt{u}})[1-q^{2}(\wh{\wt{u}}-\wt{u})(\wh{\wt{v}}-\wt{v})]^{\frac{1}{2}}+q(\wt{u}-\wh{\wt{u}})
[p^2-(\wh{u}+\wh{\wt{u}})(\wh{v}+\wh{\wt{v}})]^{\frac{1}{2}}=0, \\
& \label{dnAKNS-b} q(\wh{v}-v)[p^2-(u+\wt{u})(v+\wt{v})]^{\frac{1}{2}}+(v+\wt{v})[1-q^2(\wh{u}-u)(\wh{v}-v)]^{\frac{1}{2}}\nn \\
& \quad +(\wh{v}+\wh{\wt{v}})[1-q^2(\wh{\wt{u}}-\wt{u})(\wh{\wt{v}}-\wt{v})]^{\frac{1}{2}}+q(\wt{v}-\wh{\wt{v}})
[p^{2}-(\wh{u}+\wh{\wt{u}})(\wh{v}+\wh{\wt{v}})]^{\frac{1}{2}}=0,
\end{align}
\end{subequations}
which was introduced in \cite{Zhao-2016} by using the generalized Cauchy matrix approach \cite{ZZ-SAPM}. This is one of
the discrete versions of the first negative order AKNS equation \cite{Zhang-PD}. In the rest part of the paper,
we denote system \eqref{dnAKNS} by dAKNS(-1) for short.
The notation adopted in \eqref{dnAKNS} is as follows: both of dependent variables $u$ and $v$ are functions
defined on the two-dimensional lattice with discrete coordinates $(n,m)\in \mathbb{Z}^2$, e.g., $u=u(n, m)=:u_{n, m}$ and
the operations $u\mapsto \wt{u}$, and $u\mapsto \wh{u}$ denote elementary shifts in the two directions of the
lattice, i.e., $\wt{u}=u_{n+1,m},~\wh{u}=u_{n,m+1}$, while for the combined shift we have: $\wh{\wt{u}}=u_{n+1,m+1}$.
$p$ and $q$ are continuous lattice parameters associated with
the grid size in the directions of the lattice given by the independent variables $n$ and $m$.
Reinterpreting the variables $u$ and $v$ as $u_{n,m}=:\mathbb{U}_{n}(t)$ and $v_{n,m}=:\mathbb{V}_{n}(t)$ with $t=-\frac{m}{q}\thicksim O(1)$, then
under the continuum limit $m\rightarrow \infty,~~q\rightarrow \infty$, a semi-discrete version of
equation \eqref{dnAKNS} was also revealed, which reads
\begin{subequations}
\label{sdnAKNS}
\begin{align}
& \label{sdnAKNS-a} (\mathbb{U}'-\wt{\mathbb{U}}')(p^{2}-(\mathbb{U}+\wt{\mathbb{U}})(\mathbb{V}+\wt{\mathbb{V}}))^{\frac{1}{2}}
+(\mathbb{U}+\wt{\mathbb{U}})((1-\mathbb{U}'\mathbb{V}')^{\frac{1}{2}}+(1-\wt{\mathbb{U}}'\wt{\mathbb{V}}')^{\frac{1}{2}})=0, \\
& \label{sdnAKNS-b} (\mathbb{V}'-\wt{\mathbb{V}}')(p^{2}-(\mathbb{U}+\wt{\mathbb{U}})(\mathbb{V}+\wt{\mathbb{V}}))^{\frac{1}{2}}
+(\mathbb{V}+\wt{\mathbb{V}})((1-\mathbb{U}'\mathbb{V}')^{\frac{1}{2}}+(1-\wt{\mathbb{U}}'\wt{\mathbb{V}}')^{\frac{1}{2}})=0,
\end{align}
\end{subequations}
where the prime means the derivative of $t$. We name the equation \eqref{sdnAKNS} as sdAKNS(-1).
Following the reduction technique developed in recent papers \cite{FZ-ROMP}, we would like to consider local
and nonlocal complex reduction of the systems \eqref{dnAKNS} and \eqref{sdnAKNS}.
We call the resulting local and nonlocal complex discrete sG equation as cnd-sG, as well as
the resulting local and nonlocal complex semi-discrete sG equation as cnsd-sG.
We will construct Cauchy matrix solutions,
including soliton solutions and Jordan-block solutions, for these sG type equations.

The paper is organized as follows. In Sec. 2, we briefly recall Cauchy matrix type solutions for the dAKNS(-1) equation \eqref{dnAKNS}.
Local and nonlocal complex reduction for the dAKNS(-1) equation \eqref{dnAKNS} is investigated.
We construct some exact solutions for the cnd-sG equation. 1-soliton solution, 2-soliton solutions and
the simplest Jordan-block solution are listed. The dynamic behaviors of 1-soliton solution are analyzed and illustrated.
In Sec. 3, we consider local and nonlocal complex reduction for the sdAKNS(-1) equation \eqref{dnAKNS} and discuss the exact solutions to the resulting
local and nonlocal complex semi-discrete sG equation. Section 4 is for the conclusions.

\section{Local and nonlocal complex reduction for the dAKNS(-1) equation \eqref{dnAKNS}}\label{sec-2}

In this section, we investigate the local and nonlocal complex reduction for the dAKNS(-1) equation \eqref{dnAKNS}.
Moreover, we discuss two types of Cauchy matrix solutions for the resulting cnd-sG equation.
For the sake of simplicity, in what follows we omit the index of each unit matrix $\bI$ to indicate its size.

We start by recalling Cauchy matrix type solutions for the dAKNS(-1) equation \eqref{dnAKNS}, which were expressed as
\begin{align}
\label{Sij}
u=\tc_{2}(\bI-\bM_2\bM_1)^{-1}\br_2,\quad v=\tc_{1}(\bI-\bM_1 \bM_2)^{-1}\br_1,
\end{align}
where $\bM_1\in \mathbb{C}_{N_1\times N_2}$, $\bM_2\in \mathbb{C}_{N_2\times N_1}$,
$\br_j \in \mathbb{C}_{N_j\times 1}$, $\tc_j \in \mathbb{C}_{1\times N_j}$ with $N_1+N_2=2N$ satisfy the following determining equation set (DES)
\begin{subequations}
\label{DES-d}
\begin{align}
& \label{DES-d-M12}
\bK_1 \bM_1-\bM_1\bK_2=\br_1\, \tc_2, \quad \bK_2 \bM_2-\bM_2\bK_1=\br_2\, \tc_1, \\
& \label{DES-d-K1}
(p\bI-\bK_1)\wt{\br}_1=(p\bI+\bK_1)\br_1,\quad (p\bI+\bK_2)\wt{\br}_2=(p\bI-\bK_2)\br_2, \\
& \label{DES-d-K2}
(q\bI-\bK_1^{-1})\wh{\br}_1=(q\bI+\bK_1^{-1})\br_1,\quad (q\bI+\bK_2^{-1})\wh{\br}_2=(q\bI-\bK_2^{-1})\br_2.
\end{align}
\end{subequations}
In DES \eqref{DES-d}, $\{\br_j\}$ and $\{\bM_j\}$ are functions of $(n,m)$ while $\{\tc_j\}$ and $\{\bK_j\}$ are non-trivial constant matrices;
$\bK^{-1}_1$ and $\bK^{-1}_2$ are the inversion of matrices $\bK_1$ and $\bK_2$, respectively.
We assume that matrices $\{p\bI+(-1)^{j}\bK_j\}$ and $\{q\bI+(-1)^{j}\bK_{j}^{-1}\}$ are inverse.

Functions $u$ and $v$ are invariant and DES \eqref{DES-d} is covariant under similarity transformations
\begin{align}
\bK_j=\Ga_j^{-1}\bar{\bK}_j\Ga_j,\quad \br_j=\Ga_j^{-1}\bar{\br}_j,\quad \bs_j=\Ga^{\st}_j\bar{\bs}_j,\quad \bM_1=\Ga_1^{-1}\bar{\bM}_1\Ga_2,\quad
\bM_2=\Ga_2^{-1}\bar{\bM}_2\Ga_1,
\end{align}
where $\{\Ga_j\}$ are transform matrices. Thus solutions to the dAKNS(-1) equation \eqref{dnAKNS}
can be still given by \eqref{Sij}, where the entities satisfy canonical DES
\begin{subequations}
\label{DES-C}
\begin{align}
& \label{DES-M12-C}\Og_1 \bM_1-\bM_1\Og_2=\br_1\, \tc_2, \quad \Og_2 \bM_2-\bM_2\Og_1=\br_2\, \tc_1, \\
& \label{DES-p-C}(p\bI-\Og_1)\wt{\br}_1=(p\bI+\Og_1)\br_1,\quad (p\bI+\Og_2)\wt{\br}_2=(p\bI-\Og_2)\br_2, \\
& \label{DES-q-C}(q\bI-\Og_1^{-1})\wh{\br}_1=(q\bI+\Og_1^{-1})\br_1,\quad (q\bI+\Og_2^{-1})\wh{\br}_2=(q\bI-\Og_2^{-1})\br_2,
\end{align}
\end{subequations}
where $\Og_1$ and $\Og_2$ are the canonical forms of the matrices $\bK_1$ and $\bK_2$, respectively.

Equations \eqref{DES-C} are linear, where equations \eqref{DES-p-C} and \eqref{DES-q-C} are used to
determine plane-wave factor vector $\br$ and the equations \eqref{DES-M12-C} are used to give matrices
$\bM_1$ and $\bM_2$. The two equations in \eqref{DES-M12-C} are nothing but the famous Sylvester equations,
which have a unique solution $\{\bM_1,~\bM_2\}$ if and only if
$\mathcal{E}(\Og_1)\bigcap \mathcal{E}(\Og_2)=\varnothing$, where $\mathcal{E}(\Og_1)$ and $\mathcal{E}(\Og_2)$
represent the eigenvalue sets of $\Og_1$ and $\Og_2$, respectively. From \eqref{DES-p-C} and \eqref{DES-q-C}, we know
\begin{align}
\label{rj-solu}
\br_{j}=(p\bI+(-1)^{j-1}\Og_j)^{n}(p\bI+(-1)^{j}\Og_j)^{-n}(q\Og_j+(-1)^{j-1}\bI)^{m}
(q\Og_j+(-1)^{j}\bI)^{-m}\bC_{j},
\end{align}
where constant column vectors $\{\bC_j\}$ are phase terms of $\{\br_j\}$. The key point of
the solving procedure of \eqref{DES-M12-C} is to factorize $\bM_1$ and $\bM_2$ into triplets, i.e. $\bM_1=\bF_1\bG_1\bH_2$ and
$\bM_2=\bF_2\bG_2\bH_1$, where $\{\bF_j,~\bH_j\} \subset \mathbb{C}_{N_j\times N_j}$, $\bG_1
\in \mathbb{C}_{N_1\times N_2}$ and $\bG_2\in \mathbb{C}_{N_2\times N_1}$.
When $\{\Og_j\}$ being diagonal matrices,
one can get the soliton solutions. When $\{\Og_j\}$ being Jordan-block matrices,
multiple-pole solutions can be derived. For the detailed calculations, one can refer to \cite{Zhao-2016}.

The integrable symmetry reduction for dAKNS(-1) \eqref{dnAKNS} is of $(\sigma n,\sigma m)$ type:
\begin{align}
\label{re-dsG}
v=\delta u^*_{\sigma}, \quad \delta,~~\sigma=\pm 1,
\end{align}
where for the function $f=:f(x_1,x_2)$ we have used notation $f_{\sigma}=: f(\sigma x_1,\sigma x_2)$, which makes equations \eqref{dnAKNS-a}
and \eqref{dnAKNS-b} self-consistent leading to a single integrable equation
\begin{align}
& \label{cnd-sG}
q(\wh{u}-u)[p^2-\delta(u+\wt{u})(u^*_{\sigma}+\wt{u}^*_{\sigma})]^{\frac{1}{2}}+(u+\wt{u})[1-\delta q^2(\wh{u}-u)(\wh{u}^*_{\sigma}-u^*_{\sigma})]^{\frac{1}{2}}\nn \\
& \quad +(\wh{u}+\wh{\wt{u}})[1-\delta q^{2}(\wh{\wt{u}}-\wt{u})(\wh{\wt{u}}^*_{\sigma}-\wt{u}^*_{\sigma})]^{\frac{1}{2}}+q(\wt{u}-\wh{\wt{u}})
[p^2-\delta (\wh{u}+\wh{\wt{u}})(\wh{u}^*_{\sigma}+\wh{\wt{u}}^*_{\sigma})]^{\frac{1}{2}}=0.
\end{align}
When $\sigma=1$, equation \eqref{cnd-sG} is nothing but exactly the complex discrete sG equation.
When $\sigma=-1$, equation \eqref{cnd-sG} is referred to as a complex reverse-$(n,m)$ discrete sG equation.
We observe that equation \eqref{cnd-sG} is preserved under transformation
$u\rightarrow -u$. Besides, equation \eqref{cnd-sG} with $(\sigma,\delta)=(\pm 1,1)$ and with $(\sigma,\delta)=(\pm 1,-1)$ can be
transformed into each other by taking $u\rightarrow iu$.

To derive solutions of the cnd-sG equation \eqref{cnd-sG}, we take $N_1=N_2=N$.
For its solution, we have the following result.
\begin{Thm}
\label{so-cnd-sG-Thm}
The function
\begin{align}
\label{cnd-sG-so}
u=\tc_{2}(\bI-\bM_2\bM_1)^{-1}\br_2
\end{align}
solves the cnd-sG equation \eqref{cnd-sG}, provided that
the entities satisfy canonical DES \eqref{DES-C} and simultaneously obey the constraints
\begin{align}\label{cnd-sG-M1M12}
\br_1=\varepsilon \bT \br^*_{2,\sigma}, \quad
\tc_{1}=\varepsilon \tc^{*}_{2}\bT^{-1}, \quad
\bM_1=-\delta\sigma\bT \bM^*_{2,\sigma}\bT^*,
\end{align}
in which $\bT\in \mathbb{C}_{N\times N}$ is a constant matrix satisfying
\begin{align}\label{cnd-sG-at-eq}
\Og_1\bT+\sigma\bT\Og^*_2=0,\quad \bC_1=\varepsilon\bT\bC_2^*,\quad  \varepsilon^2=\varepsilon^{*^2}=\delta.
\end{align}
\end{Thm}
\begin{proof}

According to the assumption \eqref{cnd-sG-at-eq}, we have
\begin{align}
& \br_1 = (p\bI+\Og_{1})^{n}(p\bI-\Og_{1})^{-n}(q\Og_{1}+\bI)^{m}(q\Og_{1}-\bI)^{-m}\bC_{1} \nn \\
& \quad = \bT(p\bI-\sigma\Og^{*}_{2})^{n}(p\bI+\sigma\Og^{*}_{2})^{-n}(q\sigma{\Og_{2}^{*}}-\bI)^{m}(q\sigma{\Og_{2}^{*}}+\bI)^{-m}\bT^{-1} \bC_{1} \nn \\
& \quad = \bT(p\bI-\Og^{*}_{2})^{\sigma n}(p\bI+\Og^{*}_{2})^{-\sigma n}(q{\Og_{2}^{*}}-\bI)^{\sigma m}(q{\Og_{2}^{*}}+\bI)^{-\sigma m}\bT^{-1} \bC_{1} \nn \\
& \quad = \varepsilon \bT \br^*_{2,\sigma},
\end{align}
where we have used the identity $(a\bI+\sigma \bL)(a\bI-\sigma\bL)^{-1}=(a\bI+\bL)^{\sigma}(a\bI-\bL)^{-\sigma}$
with $a\in \mathbb{C}$. Substituting $\Og_1=-\sigma\bT\Og^*_2\bT^{-1}$ into the first equation in \eqref{DES-M12-C} and
through a straightforward calculation, we find
\begin{align}
& \Og^*_2(\sigma\bT^{-1}\bM_1\bT^{*^{-1}}+\delta\bM^*_{2,\sigma})
-(\sigma \bT^{-1}\bM_1 \bT^{*^{-1}}+\delta\bM^*_{2,\sigma})\Og^*_1=0,
\end{align}
which yields the third relation in \eqref{cnd-sG-M1M12}. With \eqref{cnd-sG-M1M12}
and \eqref{cnd-sG-at-eq} at hand, we immediately have
\begin{align*}
v=\tc_{1}(\bI-\bM_1 \bM_2)^{-1}\br_1=
\varepsilon^2 \tc_2^{*}(\bI-\bM^*_{2,\sigma}\bM^*_{1,\sigma})^{-1}\br^*_{2,\sigma}=\delta u^*_\sigma,
\end{align*}
which coincides with the reduction \eqref{re-dsG} for the cnd-sG equation \eqref{cnd-sG}.
\end{proof}

In terms of the Theorem \ref{so-cnd-sG-Thm}, we know that solution to the cnd-sG equation \eqref{cnd-sG} reads
\begin{align}
\label{cnd-sG-so-T}
u=\tc_{2}(\bI+\delta\sigma\bM_2\bT \bM^*_{2,\sigma}\bT^*)^{-1}\br_2,
\end{align}
where $\br_2$ is given by \eqref{rj-solu} and $\bM_2$ and $\bT$ are determined by
\begin{align}
\label{DES-M12-Sim}
\Og_2 \bM_2\bT+\sigma\bM_2\bT\Og^*_2=\varepsilon\br_2\, \tc^{*}_{2}.
\end{align}
We denote $\bM_2\bT\rightarrow\dc{\bM}_2$ and simplify solution \eqref{cnd-sG-so-T} together with Sylvester equation
\eqref{DES-M12-Sim} as
\begin{subequations}
\label{cnd-sG-so-nT}
\begin{align}
& \label{cnd-sG-so-nT-a}
u=\tc_{2}(\bI+\delta\sigma\dc{\bM}_2 \dc{\bM}^*_{2,\sigma})^{-1}\br_2, \\
& \label{KM-rtc} \Og_2 \dc{\bM}_2+\sigma\dc{\bM}_2\Og^*_2=\varepsilon\br_2\, \tc^{*}_{2}.
\end{align}
\end{subequations}
As three examples, we just list
1-soliton solution, 2-soliton solutions and the simplest Jordan-block solution.
For the sake of brevity, we introduce two notations
\begin{align}
\label{nota}
\alpha_{ij}=k_i+\sigma k_j^{*},\quad  \rho_j=\rho_j(n,m)=\bigg(\frac{p-k_j}{p+k_j}\bigg)^{n}\bigg(\frac{qk_j-1}{qk_j+1}\bigg)^{m}\rho_j^0,
\end{align}
as well as $\alpha=\alpha_{11}|_{k_1\rightarrow k}$ and $\rho=\rho_1|_{k_1\rightarrow k,\rho_1^0\rightarrow \rho^0}$,
where $\{k_j,~\rho_j^0,~k,~\rho^0\}$ are complex constants. $\rho_j$ and $\rho$ play the role of discrete plane-wave factors.

When $N=1$, we denote
\begin{align}
\label{K-tc-1}
\Og_2=k_1, \quad \tc_{2}=c_1,
\end{align}
and write down 1-soliton solution
\begin{align}
\label{u-soli-1}
u=\frac{c_1\alpha_{11}^2\rho_1}{\alpha_{11}^2+\delta |c_1|^2\rho_1\rho^*_{1,\sigma}}
\end{align}
with module $|\cdot|$. When $N=2$, we take
\begin{align}
\label{K-tc-2}
\Og_2=\text{diag}(k_1,k_2), \quad \tc_{2}=(c_1,c_2).
\end{align}
In this case, the 2-soliton solutions
are expressed as $u=\dfrac{g_2}{g_1}$, where
\begin{subequations}
\begin{align}
& g_1=1+\delta\sum_{i=1}^2\sum_{j=1}^2
\left(\frac{c_ic_j^*\rho_i\rho^*_{j,\sigma}}{\alpha^2_{ij}}\right)
+\frac{|c_1c_2|^2|k_1-k_2|^4\rho_1\rho_2\rho^*_{1,\sigma}\rho^*_{2,\sigma}}
{\alpha^2_{11}\alpha^2_{12}\alpha^2_{21}\alpha^2_{22}}, \\
& g_2=c_1\rho_1+c_2\rho_2+\delta c_1c_2(k_1-k_2)^2\rho_1\rho_2
\bigg(\frac{c^*_1 \rho^*_{1,\sigma}}{\alpha^2_{11}\alpha^2_{21}}+\frac{
c^*_2\rho^*_{2,\sigma}}{\alpha^2_{12}\alpha^2_{22}}\bigg).
\end{align}
\end{subequations}
For presenting the simplest Jordan-block solution, we set $N=2$ together with
\begin{align}
\label{K2-tc-def-Jor}
\Og_2=\left(\begin{array}{cc}
k & 0   \\
1   & k
\end{array}\right),\quad \tc_{2}=(c_1,c_2).
\end{align}
Then we have
\begin{subequations}
\label{rJ-solu}
\begin{align}
\br_2=(\rho,\cd{\rho})^{\st}, \quad \dc{\bM}_2=\bF\bG\bH,
\end{align}
in which
\begin{align}
& \bF=\left(\begin{array}{cc}
\rho & 0   \\
\cd{\rho}  & \rho
\end{array}\right), \quad
\bG=\left(\begin{array}{cc}
\alpha^{-1} & \alpha^{-2}    \\
-\alpha^{-2}  & -2\alpha^{-3}
\end{array}\right), \quad
\bH=\left(\begin{array}{cc}
c_1 & c_2 \\
c_2  & 0
\end{array}\right),
\end{align}
\end{subequations}
where $\cd{\rho}=:\partial_k \rho$. Substituting \eqref{K2-tc-def-Jor} and \eqref{rJ-solu}
into \eqref{cnd-sG-so-nT-a}, we get the simplest Jordan-block solution
$u=\dfrac{h_2}{h_1}$ with
\begin{subequations}
\begin{align}
& h_1=\alpha^{8}+|c_{2}|^{4}\rho^{2}{\rho_{\sigma}^{*}}^{2}+\sigma\delta\big[c_{2}^{*}\alpha^{5}\cd{\rho}
(c_{2}\sigma\alpha\cd{\rho}^{*}_{\sigma}+(c_{1}\sigma\alpha+c_{2}\beta)\rho^{*}_{\sigma}) \nn \\
&\qquad +\alpha^{4}\rho\big(\rho^{*}_{\sigma}\big(c_{1}\alpha\zeta+c_{2}(c_{1}^{*}\alpha\beta
-2c_{2}^{*}(\sigma+\beta))\big)+c_{2}\alpha\zeta\cd{\rho}_{\sigma}^{*}\big)\big], \\
& h_2=\alpha^{3}\big[\alpha^{5}(c_{1}\rho+c_{2}\cd{\rho})+\sigma\delta\rho^{2}\big(c_{2}\alpha(|c_{2}|^{2}-\sigma\alpha\gamma)\cd{\rho}_{\sigma}^{*} \nn \\
&\qquad -(\alpha(c_{1}\sigma\alpha\gamma+c_{2}\beta(c_{1}c_{2}^{*}+\gamma))
+c_{2}^{2}(2c_{2}^{*}\beta-c_{1}^{*}\sigma\alpha))\rho^{*}_{\sigma}\big)\big],
\end{align}
\end{subequations}
where $\beta=\sigma-1,~~\gamma=c_{1}c_{2}^{*}-c_{1}^{*}c_{2},~~
\zeta=c_{1}^{*}\sigma \alpha-c_{2}^{*}\beta$.

To understand the dynamic behavior of soliton solution \eqref{u-soli-1}, we set
\begin{align}
\label{decom}
k_1=\mu+i\nu, \quad \text{with} \quad \mu,~~\nu \in \mathbb{R}.
\end{align}
For $\sigma=1$, i.e., local case, the carrier wave is expressed as
\begin{align}
\label{u-d-ca-I}
|u|^2=\Biggl\{
\begin{array}{ll}
\mu^2\sech^2(\frac{1}{2}\ln A+\ln\frac{|c_1\rho_1^0|}{2|\mu|}),& \text{with} \quad \delta=1, \\
\mu^2\csch^2(\frac{1}{2}\ln A+\ln\frac{|c_1\rho_1^0|}{2|\mu|}),& \text{with} \quad \delta=-1,\\
\end{array}
\end{align}
where $A=\left(\dfrac{(p-\mu)^{2}+\nu^{2}}{(p+\mu)^{2}+\nu^{2}}\right)^{n}\left( \dfrac{(q\mu-1)^{2}+q^{2}\nu^{2}}{(q\mu+1)^{2}+q^{2}\nu^{2}}\right)^{m}$.
As $\delta=1$, the solution \eqref{u-d-ca-I} is nonsingular and provides a bell-type traveling wave which
propagates with initial phase $\ln\frac{|c_1\rho_1^0|}{2|\mu|}$. The amplitude is approximately equal to
$\mu^2$. As $\delta=-1$, the solution \eqref{u-d-ca-I} has large value in the neighborhood of straight line
\begin{align}
\frac{y}{2}\ln \frac{(p-\mu)^2+\nu^2}{(p+\mu)^2+\nu^2}
+\frac{z}{2}\ln \frac{(q\mu-1)^2+(q\nu)^2}{(q\mu+1)^2+(q\nu)^2}+\ln\frac{|c_1\rho_1^0|}{2|\mu|}=0, \quad y,z\in \mathbb{R}.
\end{align}
We illustrate these two solitons in Figure 1.

\begin{center}
\begin{picture}(120,100)
\put(-120,-23){\resizebox{!}{4cm}{\includegraphics{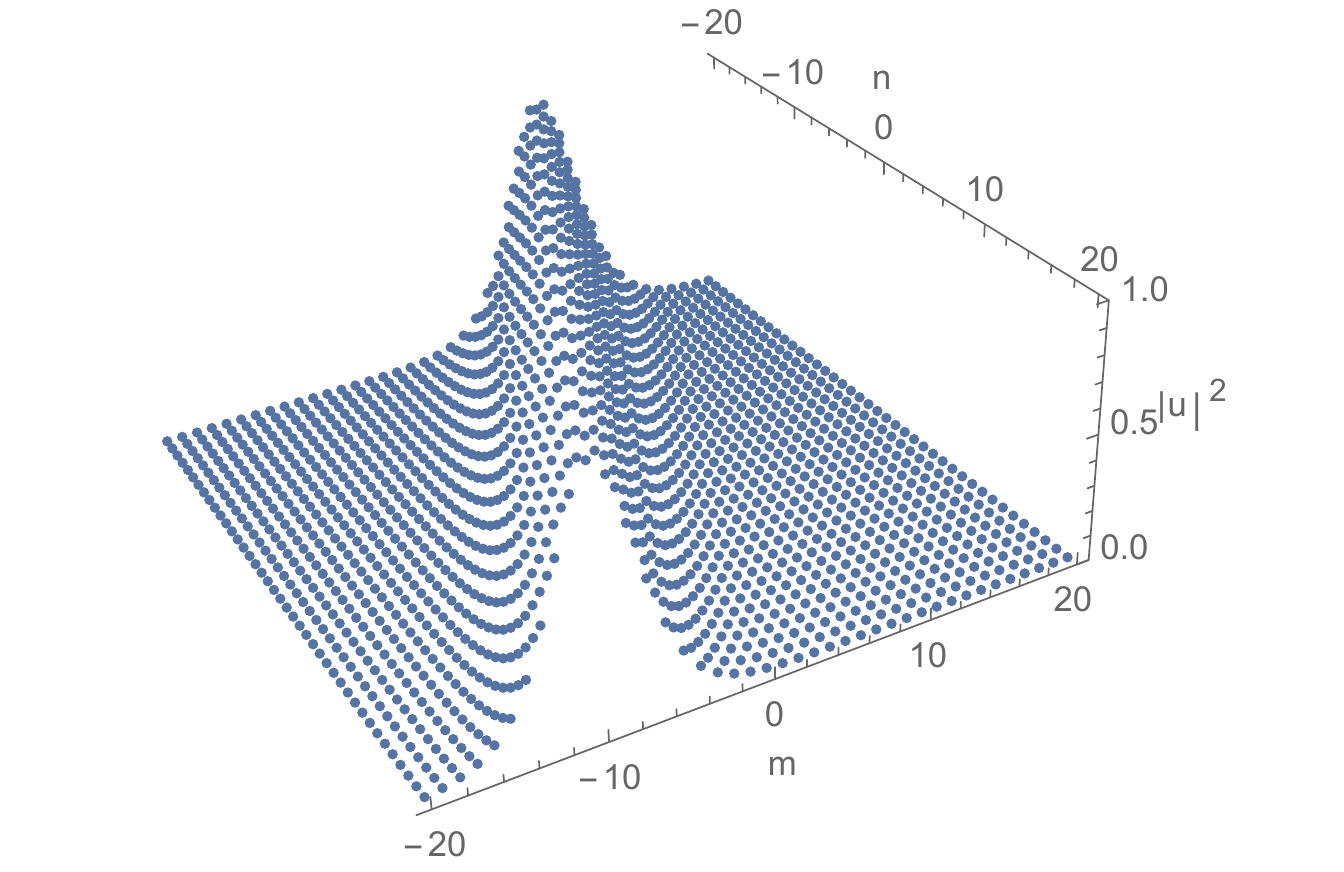}}}
\put(100,-23){\resizebox{!}{4cm}{\includegraphics{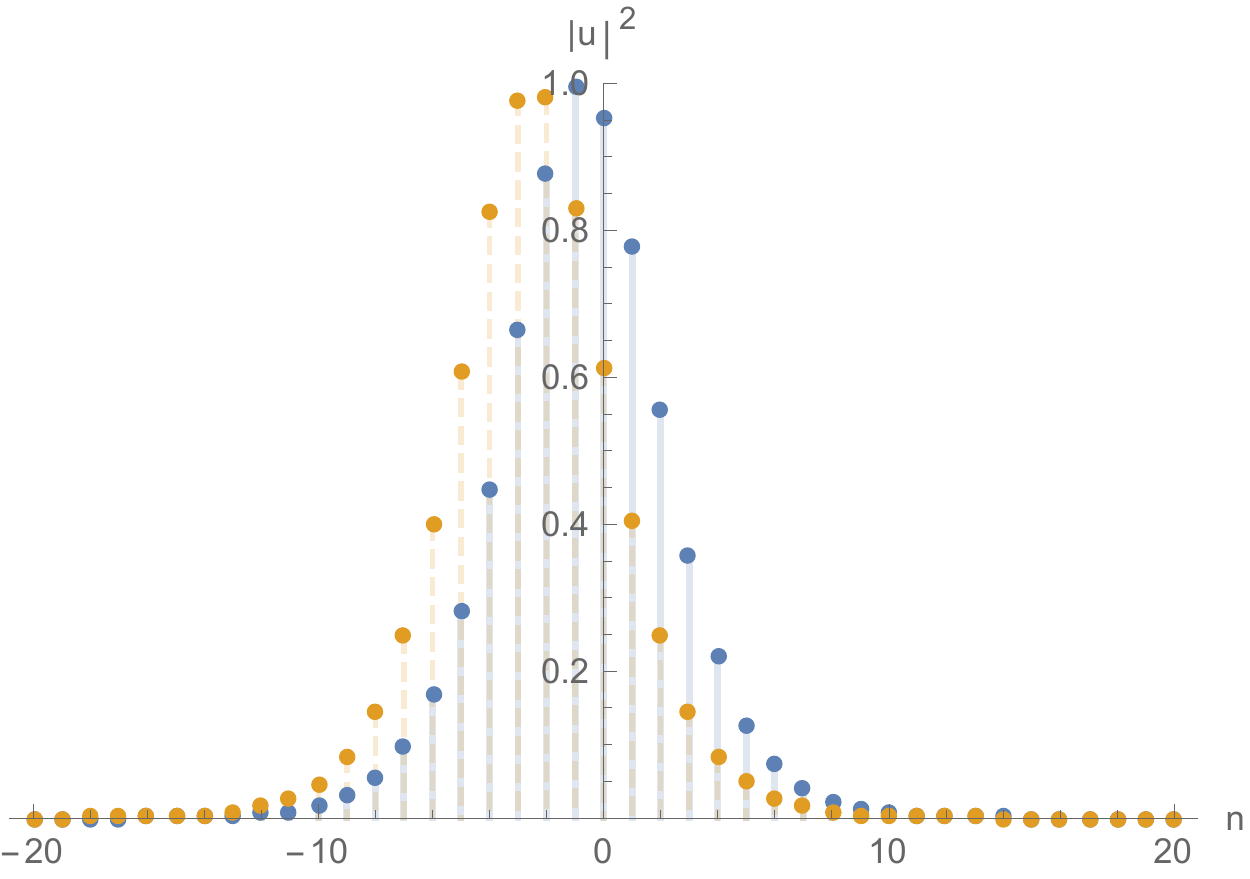}}}
\end{picture}
\end{center}
\vskip 20pt
\begin{center}
\begin{minipage}{15cm}{\footnotesize
\quad\qquad\qquad\qquad\qquad\qquad(a)\qquad\qquad\qquad\qquad\qquad\qquad\qquad\qquad\qquad\qquad\quad\quad\quad (b)}
\end{minipage}
\end{center}
\vskip 10pt
\begin{center}
\begin{picture}(120,80)
\put(-120,-23){\resizebox{!}{4cm}{\includegraphics{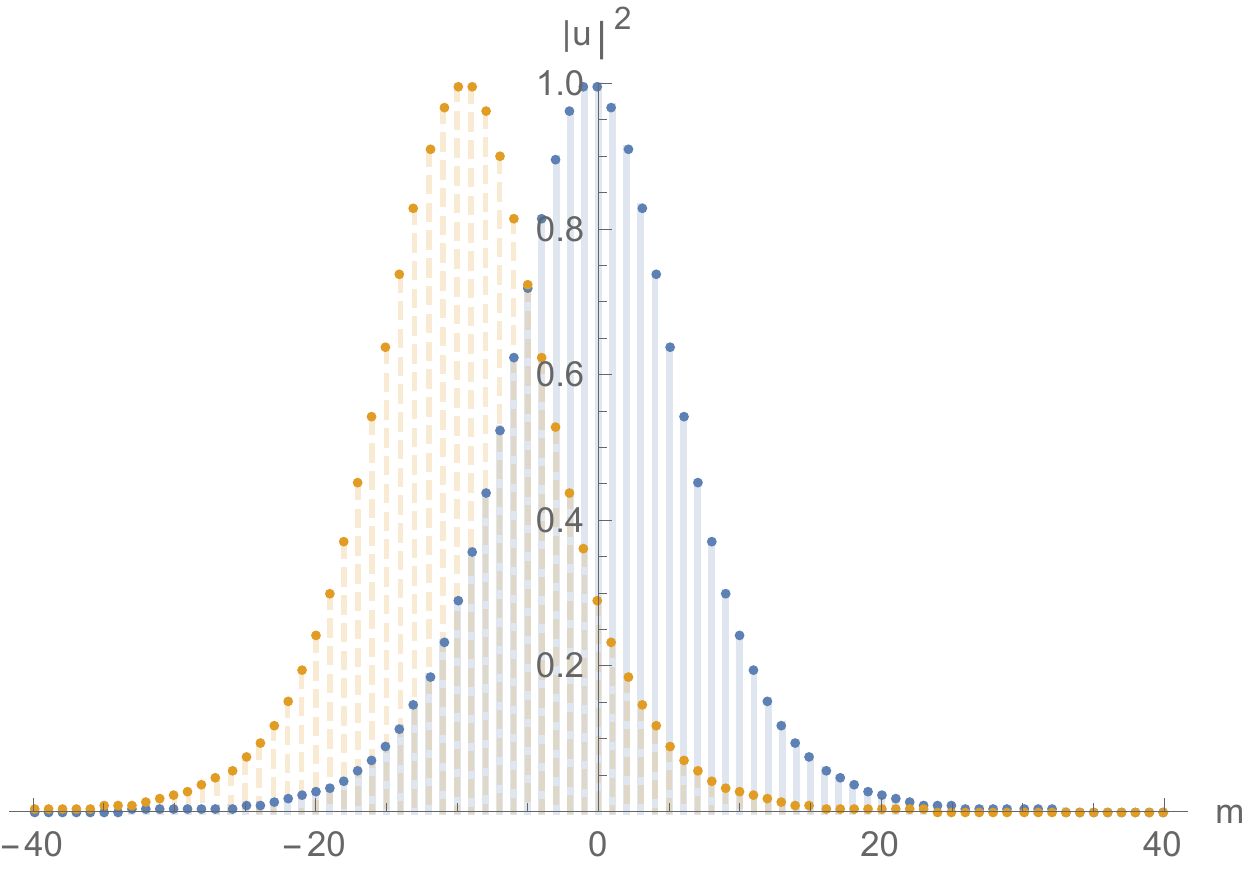}}}
\put(100,-23){\resizebox{!}{4cm}{\includegraphics{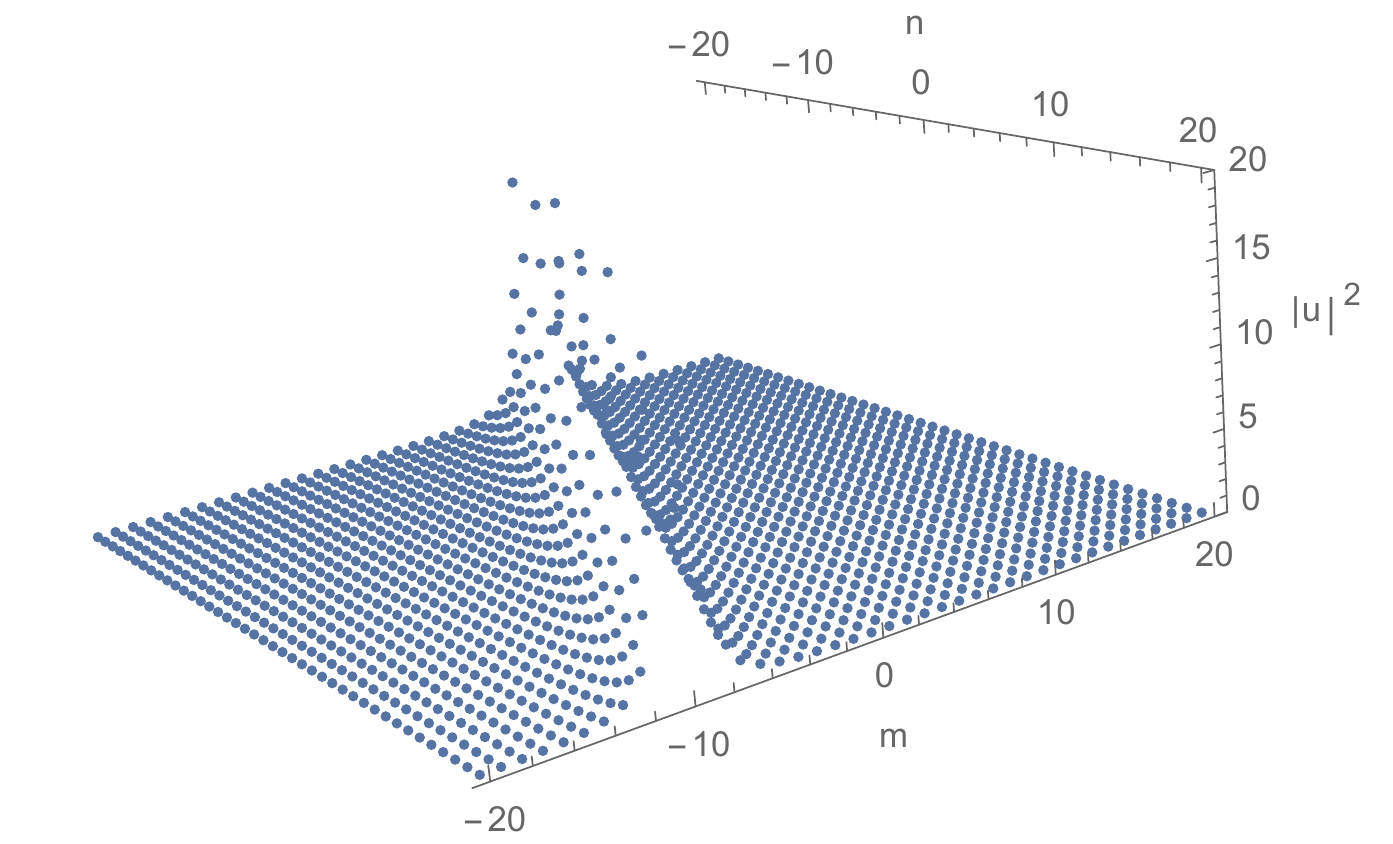}}}
\end{picture}
\end{center}
\vskip 20pt
\begin{center}
\begin{minipage}{15cm}{\footnotesize
\quad\qquad\qquad\qquad\qquad\qquad(c)\qquad\qquad\qquad\qquad\qquad\qquad\qquad\qquad\qquad\qquad\quad\quad\quad (d)\\
{\bf Fig. 1} shape and motion with $|u|^2$ given by \eqref{u-d-ca-I} for $k_1=1+3i, p=2, q=1.5, \rho_1^0=1$ and $c_1=1+i$.
(a) 3D-plot for $\delta=1$.
(b) waves in blue and yellow stand for plot (a) at $m=-1$ and $m=3$, respectively.
(c) waves in blue and yellow stand for plot (a) at $n=-1$ and $n=3$, respectively.
(d) 3D-plot for $\delta=-1$.}
\end{minipage}
\end{center}

For $\sigma=-1$, i.e., nonlocal case, the wave package reads
\begin{align}
\label{u-d-ca-II}
|u|^2=\dfrac{4\nu^{2}A}{B+B^{-1}-2\delta CA^{-1}\cos(2n\arctan \theta_1+2m\arctan\theta_2)},
\end{align}
where $A$ is defined by \eqref{u-d-ca-I} and
\begin{subequations}
\begin{align}
& \label{B-def} B=4\nu^{2}|c_1\rho_1^0|^{-2}, \quad \theta_1=\frac{2 p\nu}{\mu^{2}+\nu^{2}-p^{2}}, \quad
\theta_2=\frac{2q\nu}{q^{2}(\mu^{2}+\nu^{2})-1}, \\
& C=\dfrac{((p^{2}-\mu^{2}-\nu^{2})^{2}+4p^{2}\nu^{2})^{n}((q^{2}(\mu^{2}+\nu^{2})-1)^{2}+4q^{2}\nu^{2})^{m}}
{((p+\mu)^{2}+\nu^{2})^{2n}((q\mu+1)^{2}+q^{2}\nu^{2})^{2m}}.
\end{align}
\end{subequations}
The solution \eqref{u-d-ca-II} has oscillatory phenomenon since
the involvement of cosine function in denominator.
We illustrate soliton \eqref{u-d-ca-II} in Figure 2.

\begin{center}
\begin{picture}(120,100)
\put(-120,-23){\resizebox{!}{4cm}{\includegraphics{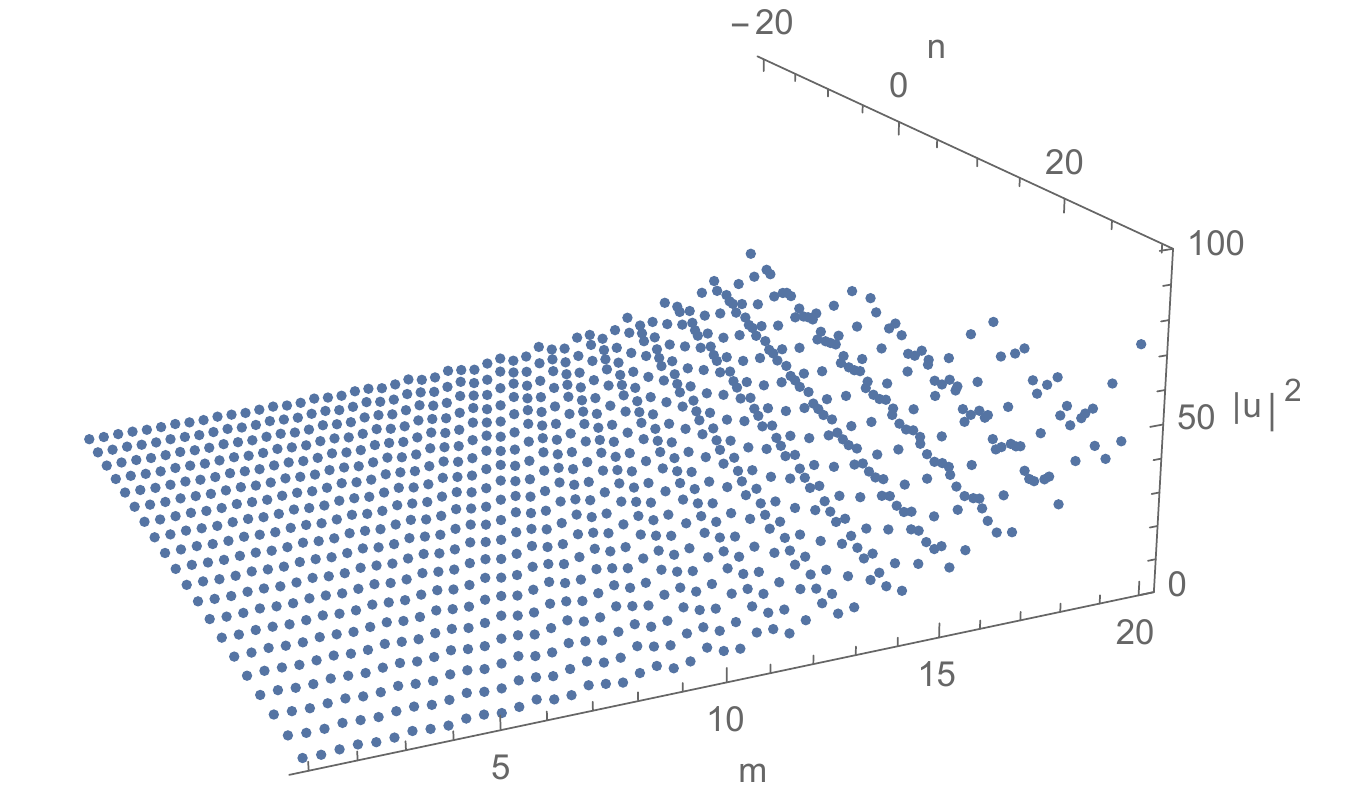}}}
\put(100,-23){\resizebox{!}{4cm}{\includegraphics{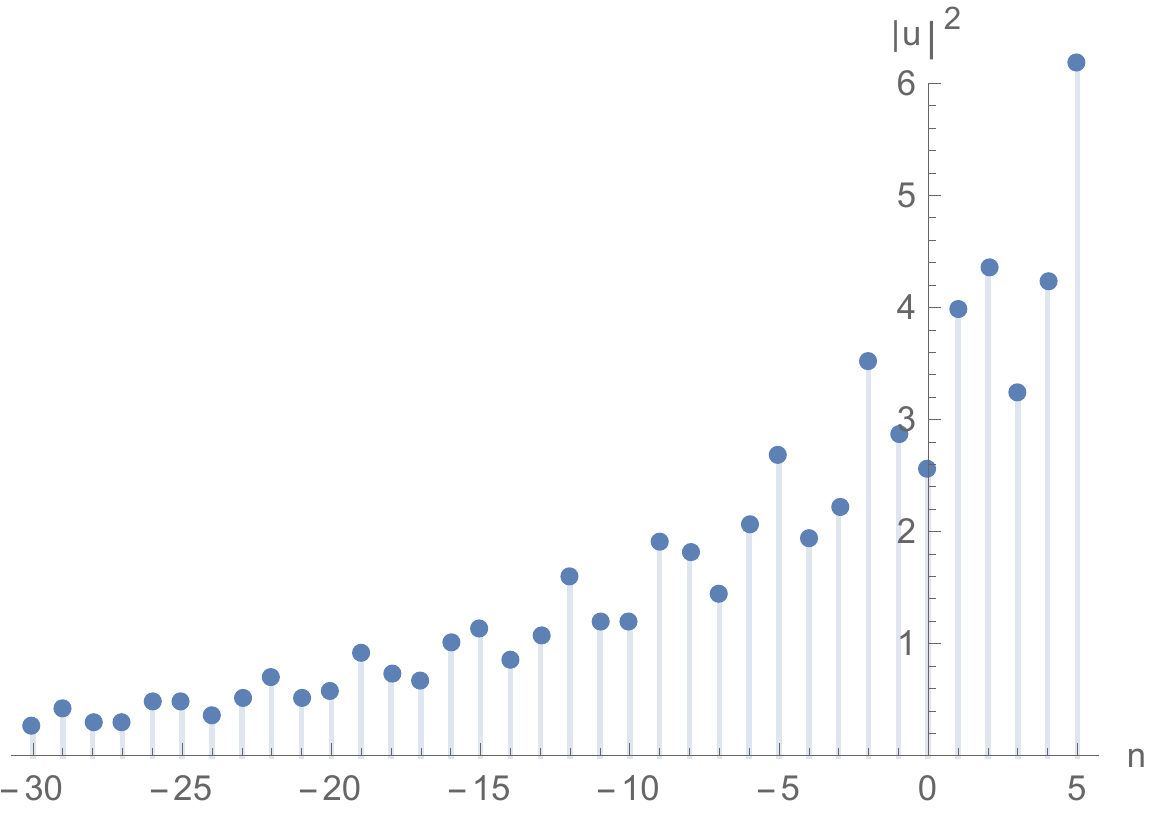}}}
\end{picture}
\end{center}
\vskip 20pt
\begin{center}
\begin{minipage}{15cm}{\footnotesize
\quad\qquad\qquad\qquad\qquad\qquad(a)\qquad\qquad\qquad\qquad\qquad\qquad\qquad\qquad\qquad\qquad\quad \qquad \quad (b)}
\end{minipage}
\end{center}
\vskip 10pt
\begin{center}
\begin{picture}(120,80)
\put(-120,-23){\resizebox{!}{4cm}{\includegraphics{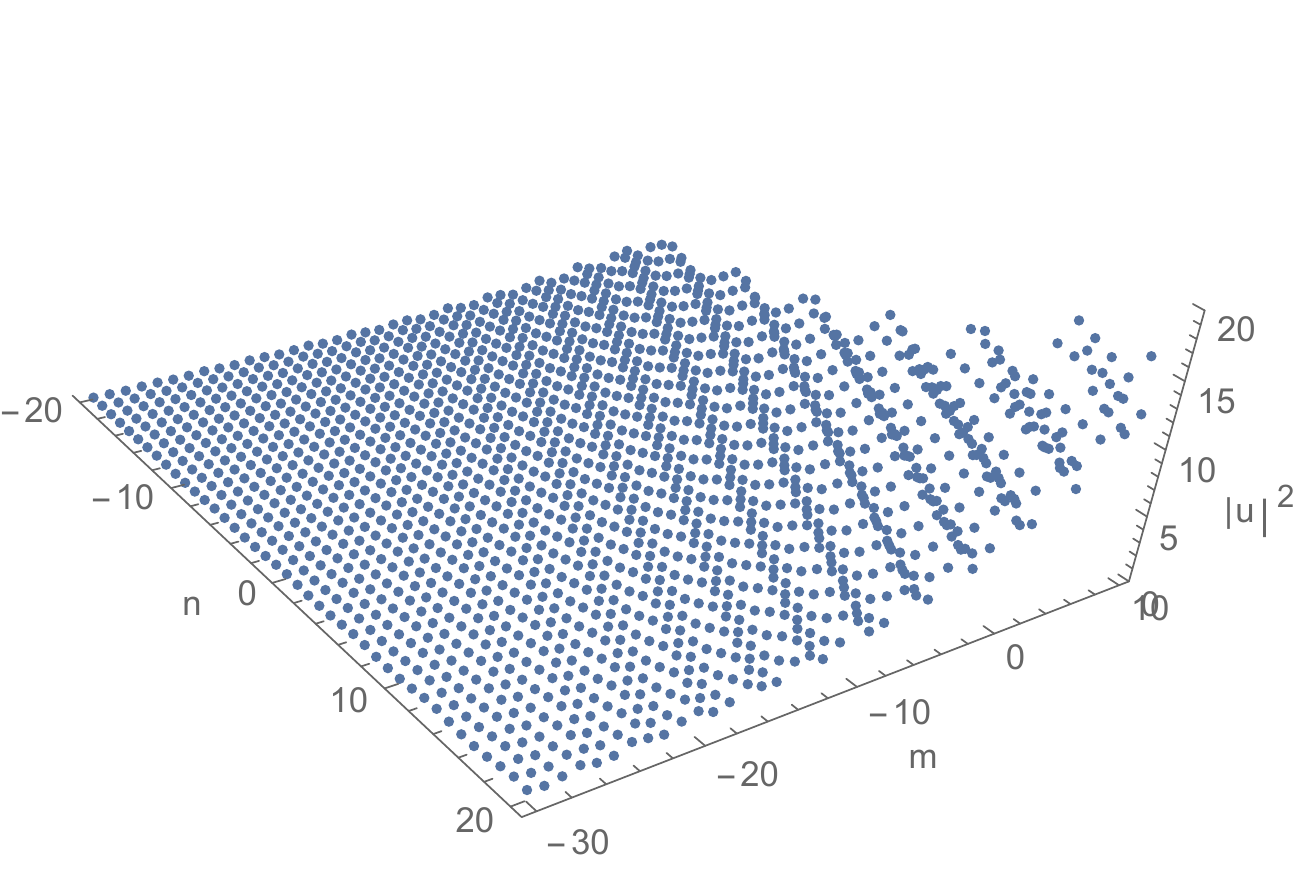}}}
\put(100,-23){\resizebox{!}{4cm}{\includegraphics{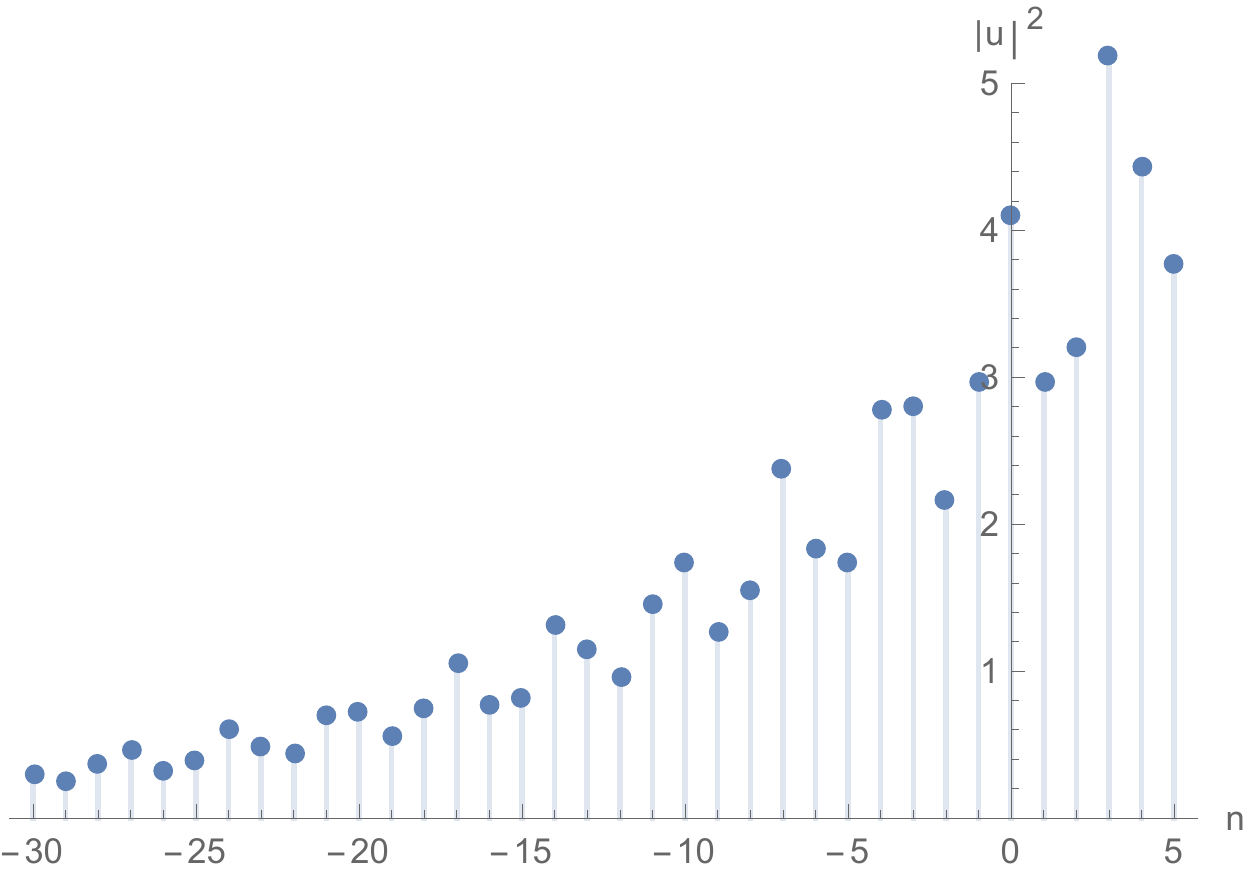}}}
\end{picture}
\end{center}
\vskip 20pt
\begin{center}
\begin{minipage}{15cm}{\footnotesize
\quad\qquad\qquad\qquad\qquad\qquad(c)\qquad\qquad\qquad\qquad\qquad\qquad\qquad\qquad\qquad\qquad\quad \qquad \quad (d)\\
{\bf Fig. 2} shape and motion with $|u|^2$ given by \eqref{u-d-ca-II} for $k_1=-0.1+2i, p=1, q=2, \rho_1^0=1$ and $c_1=1+i$.
(a) 3D-plot for $\delta=1$. (b) wave of plot (a) at $m=10$.
(c) 3D-plot for $\delta=-1$.
(d) wave of plot (c) at $m=10$.}
\end{minipage}
\end{center}

\section{Local and nonlocal complex reduction for the sdAKNS(-1) equation \eqref{sdnAKNS}}\label{sec-3}

In this section, we shall use a similar strategy to consider the local and nonlocal complex reduction for the sdAKNS(-1) equation \eqref{sdnAKNS}.

We notice that
\begin{align}
\label{UV-solu}
\mathbb{U}=\tc_{2}(\bI-\bM_2\bM_1)^{-1}\br_2,\quad \mathbb{V}=\tc_{1}(\bI-\bM_1 \bM_2)^{-1}\br_1,
\end{align}
solve the sdAKNS(-1) equation \eqref{sdnAKNS} \cite{Zhao-2016}, where
the components $\{\bM_1,~\bM_2,~\br_j,~\tc_j\}$ satisfy the equations \eqref{DES-M12-C},
\eqref{DES-p-C} together with
\begin{align}
\label{DES-sd-r12}
\br'_1=-2\Og_{1}^{-1}\br_1,\quad \br'_2=2\Og_{2}^{-1}\br_2.
\end{align}

Under the complex integrable symmetry reduction
\begin{align}
\label{re-sdsG}
\mathbb{V}=\delta \mathbb{U}^*_{\sigma}, \quad \delta,~~\sigma=\pm 1,
\end{align}
the coupled equations \eqref{sdnAKNS-a} and \eqref{sdnAKNS-b} are compatible and give rise to the cnsd-sG equation
\begin{align}
\label{cnsd-sG}
(\mathbb{U}'-\wt{\mathbb{U}}')(p^{2}-\delta(\mathbb{U}+\wt{\mathbb{U}})(\mathbb{U}^*_{\sigma}+\wt{\mathbb{U}}^*_{\sigma}))^{\frac{1}{2}}
+(\mathbb{U}+\wt{\mathbb{U}})((1-\delta \mathbb{U}'(\mathbb{U}_{\sigma}^{*})')^{\frac{1}{2}}
+(1-\delta\wt{\mathbb{U}}'(\wt{\mathbb{U}}_{\sigma}^{*})')^{\frac{1}{2}})=0,
\end{align}
which is the complex semi-discrete sG equation as $\sigma=1$, respectively,
the complex reverse-$(n,t)$ semi-discrete sG equation as $\sigma=-1$. Similar to the
equation \eqref{cnd-sG}, the cnsd-sG equation \eqref{cnsd-sG}
is preserved under transformation $u\rightarrow -u$. And equation \eqref{cnsd-sG} with $(\sigma,\delta)=(\pm 1,1)$ and with $(\sigma,\delta)=(\pm 1,-1)$ can be
transformed into each other by taking $u\rightarrow iu$.

Under the assumption $N_1=N_2=N$, solution to the equation \eqref{cnsd-sG} can be summarized by the following theorem, where
we skip the proof because it is very similar to the cnd-sG case.
\begin{Thm}
\label{so-cnsd-sG-Thm}
The function
\begin{align}
\label{cnsd-sG-so}
\mathbb{U}=\tc_{2}(\bI-\bM_2\bM_1)^{-1}\br_2,
\end{align}
solves the cnsd-sG equation \eqref{cnsd-sG}, provided that
the entities satisfy DES \eqref{DES-M12-C},
\eqref{DES-p-C}, \eqref{DES-sd-r12} and simultaneously obey the constraints
\eqref{cnd-sG-M1M12} and \eqref{cnd-sG-at-eq}.
\end{Thm}
We find that solution for the cnsd-sG equation \eqref{cnsd-sG} is expressed by
\begin{align}
\label{cnsd-sG-so-nT-a}
\mathbb{U}=\tc_{2}(\bI+\delta\sigma\dc{\bM}_2 \dc{\bM}^*_{2,\sigma})^{-1}\br_2,
\end{align}
in which the entities still satisfy the Sylvester equation \eqref{KM-rtc} but with
\begin{align}
\label{r2-solu}
\br_2=(p\bI-\Og_2)^{n}(p\bI+\Og_2)^{-n}\mbox{exp}(2\Og_{2}^{-1}t)\bD_2,
\end{align}
where $\bD_2$ is a $N$-th order constant column vector.

With \eqref{K-tc-1} we identify the 1-soliton solution to equation \eqref{cnsd-sG}
\begin{align}
\label{u-soli-sd}
\mathbb{U}=\frac{c_1 \alpha_{11}^2\varrho_1}{\alpha_{11}^2+\delta |c_1|^2\varrho_1\varrho^*_{1,\sigma}},
\end{align}
where the semi-discrete plane-wave factor $\varrho_1$ is given by $\varrho_1=\big(\frac{p-k_1}{p+k_1}\big)^{n}\exp(\frac{2t}{k_1})\varrho_1^0$.
Solution \eqref{u-soli-sd} can be also obtained from \eqref{u-soli-1} by replacing $\rho_1$ by $\varrho_1$.
This operation is also valid for the 2-soliton solutions and the simplest Jordan-block solution. Here we skip the
explicit expressions for these two solutions.


With the help of decomposition \eqref{decom}, the carrier wave of \eqref{u-soli-sd} with $\sigma=1$ is expressed as
\begin{align}
\label{u-d-ca-P}
|\mathbb{U}|^2=\Biggl\{
\begin{array}{ll}
\mu^2\sech^2(\frac{1}{2}\ln D+\ln\frac{|c_1\varrho_1^0|}{2|\mu|}),& \text{with} \quad \delta=1, \\
\mu^2\csch^2(\frac{1}{2}\ln D+\ln\frac{|c_1\varrho_1^0|}{2|\mu|}),& \text{with} \quad \delta=-1, \\
\end{array}
 \end{align}
where $D=\left(\dfrac{(p-\mu)^{2}+\nu^{2}}{(p+\mu)^{2}+\nu^{2}}\right)^n\exp(\frac{4\mu t}{\mu^2+\nu^2})$.
As $\delta=1$, the solution \eqref{u-d-ca-P} is nonsingular and the wave
propagates with initial phase $\ln\frac{|c_1\varrho_1^0|}{2|\mu|}$, amplitude $\mu^2$,
top trajectory
\begin{align}
\label{tt}
\frac{n}{2}\ln \frac{(p-\mu)^2+\nu^2}{(p+\mu)^2+\nu^2}
+\dfrac{2\mu t}{\mu^{2}+\nu^{2}}+\ln\frac{|c_1\varrho_1^0|}{2|\mu|}=0,
\end{align}
and velocity $n'(t)=\dfrac{4\mu }{\mu^{2}+\nu^{2}}\ln^{-1} \frac{(p+\mu)^2+\nu^2}{(p-\mu)^2+\nu^2}$.
As $\delta=-1$, the solution \eqref{u-d-ca-P} has singularities along with point trace \eqref{tt}.
We illustrate these two solitons in Figure 3.

\begin{center}
\begin{picture}(120,100)
\put(-120,-23){\resizebox{!}{4cm}{\includegraphics{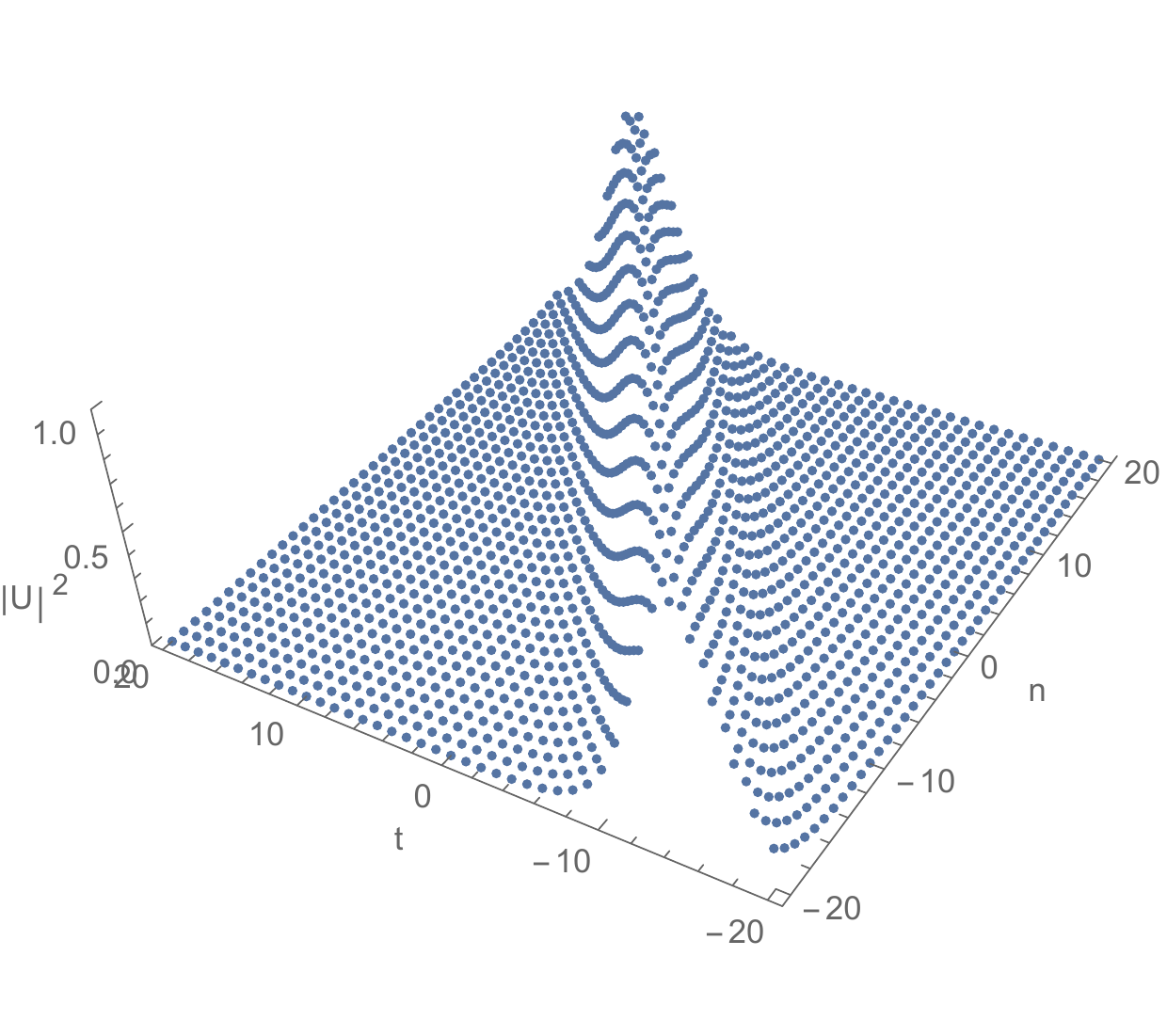}}}
\put(100,-23){\resizebox{!}{4cm}{\includegraphics{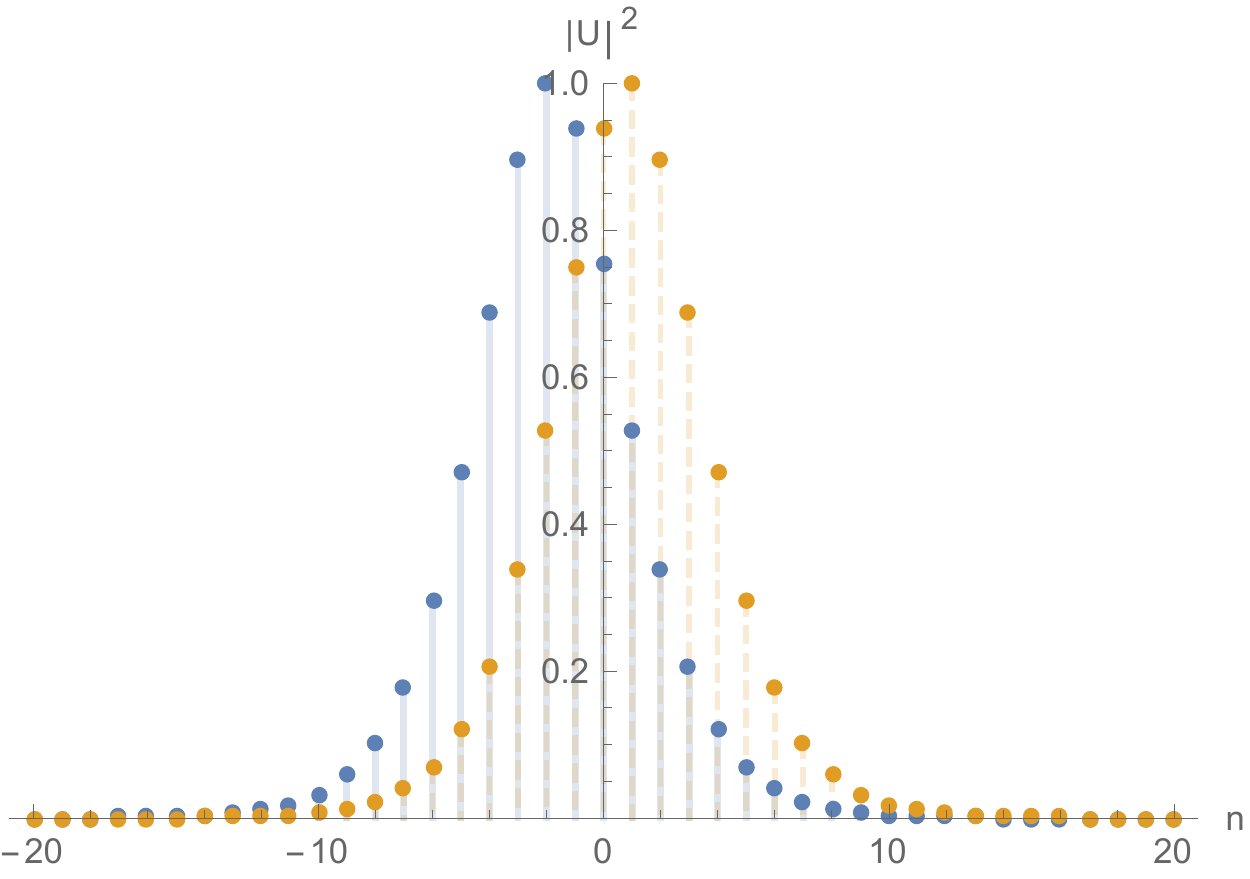}}}
\end{picture}
\end{center}
\vskip 20pt
\begin{center}
\begin{minipage}{15cm}{\footnotesize
\quad\qquad\qquad\qquad\qquad\qquad(a)\qquad\qquad\qquad\qquad\qquad\qquad\qquad\qquad\qquad\qquad\quad \qquad \quad (b)}
\end{minipage}
\end{center}
\vskip 10pt
\begin{center}
\begin{picture}(120,80)
\put(-120,-23){\resizebox{!}{4cm}{\includegraphics{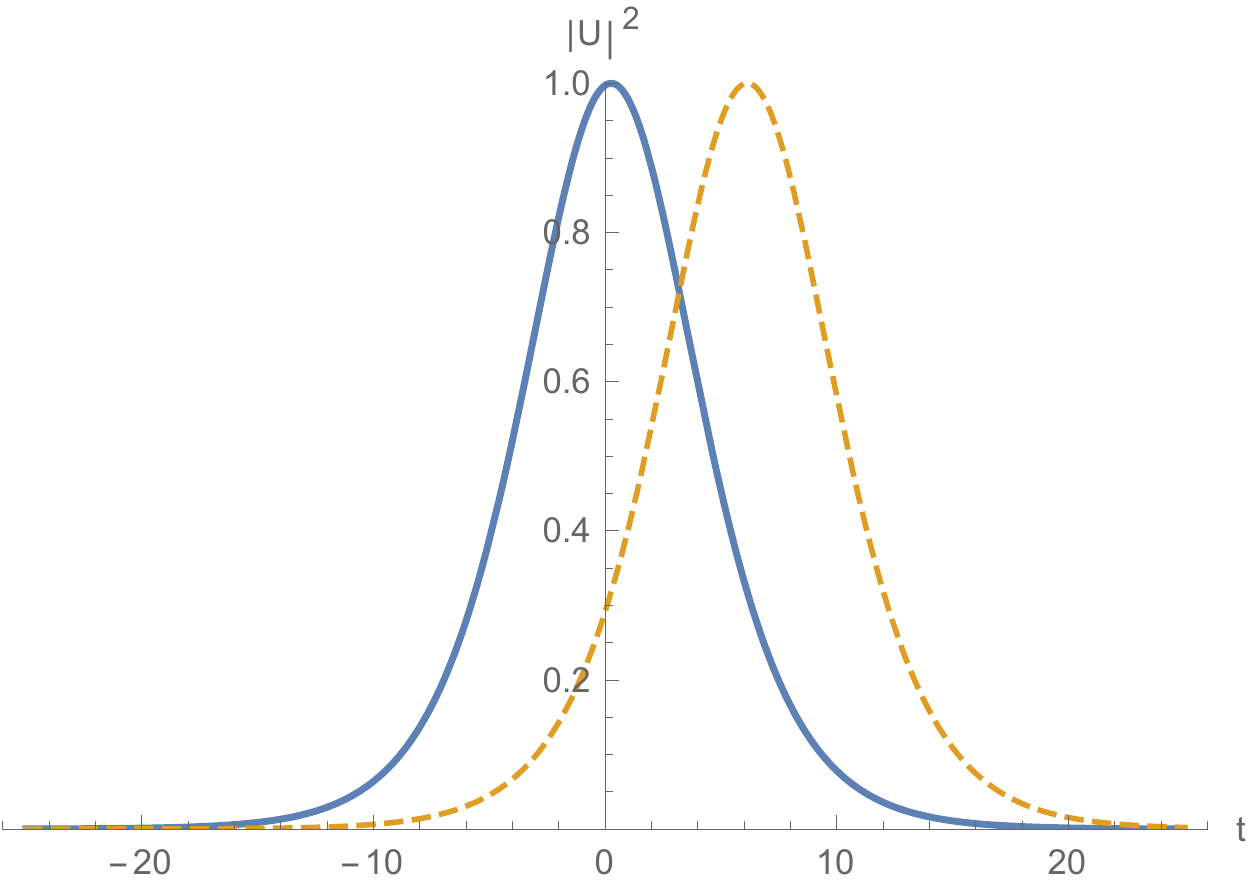}}}
\put(60,-23){\resizebox{!}{4cm}{\includegraphics{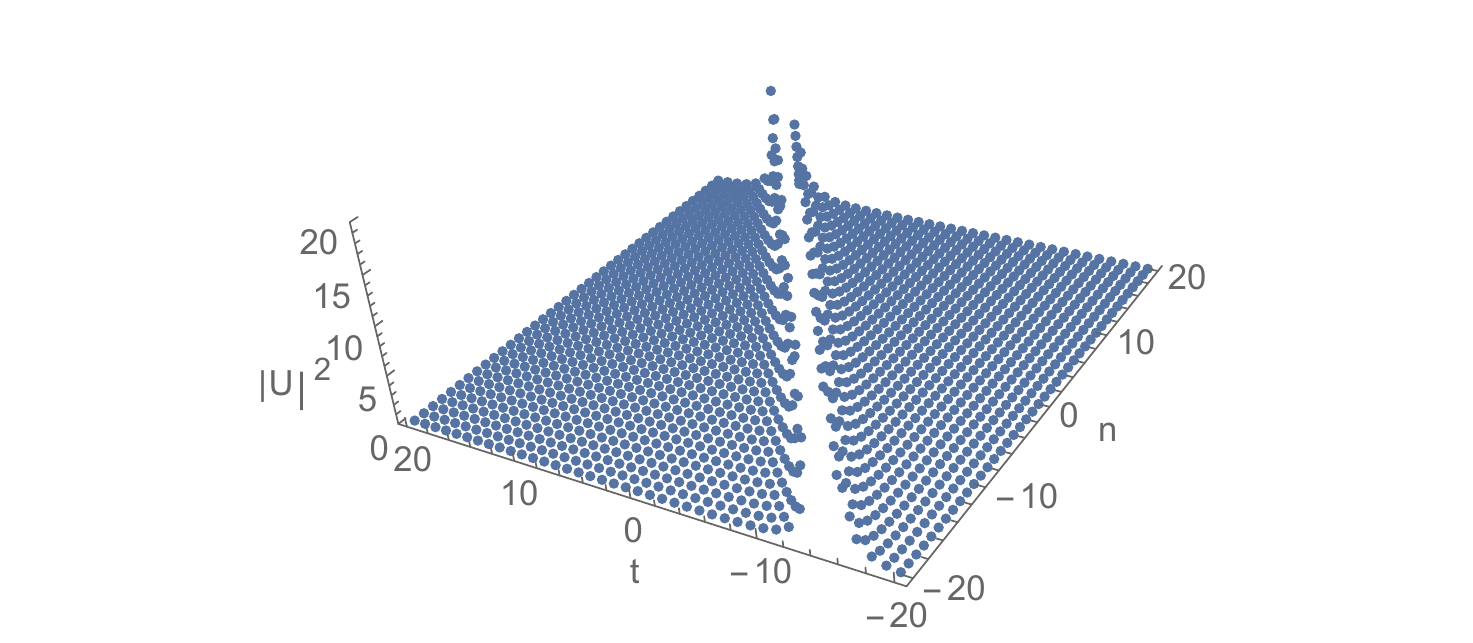}}}
\end{picture}
\end{center}
\vskip 20pt
\begin{center}
\begin{minipage}{15cm}{\footnotesize
\quad\qquad\qquad\qquad\qquad\qquad(c)\qquad\qquad\qquad\qquad\qquad\qquad\qquad\qquad\qquad\quad\qquad  \qquad \quad (d)\\
{\bf Fig. 3} shape and motion with $|\mathbb{U}|^2$ given by \eqref{u-d-ca-P} for $k_1=1+3i, p=2, \varrho_1^0=1$ and $c_1=1+i$.
(a) 3D-plot for $\delta=1$.
(b) waves in blue and yellow stand for plot (a) at $t=-1$ and $t=3$, respectively.
(c) waves in blue and yellow stand for plot (a) at $n=-1$ and $n=3$, respectively.
(d) 3D-plot for $\delta=-1$.}
\end{minipage}
\end{center}

For $\sigma=-1$, the wave package reads
\begin{align}
\label{u-d-ca-F}
|\mathbb{U}|^2=\dfrac{4\nu^{2}D}{B+B^{-1}-2\delta ED^{-1}\cos(2n\arctan\theta_1-4\nu t)},
\end{align}
where $B$ is the same as \eqref{B-def} and $E=\dfrac{((p-\mu^{2}-\nu^{2})^{2}+4p^{2}\nu^{2})^{n}}{((p+\mu)^{2}+\nu^{2})^{2n}}\exp(\frac{4\mu t}{\mu^{2}+\nu^{2}})$.
We illustrate this wave in Figure 4.

\begin{center}
\begin{picture}(120,100)
\put(-120,-23){\resizebox{!}{4.5cm}{\includegraphics{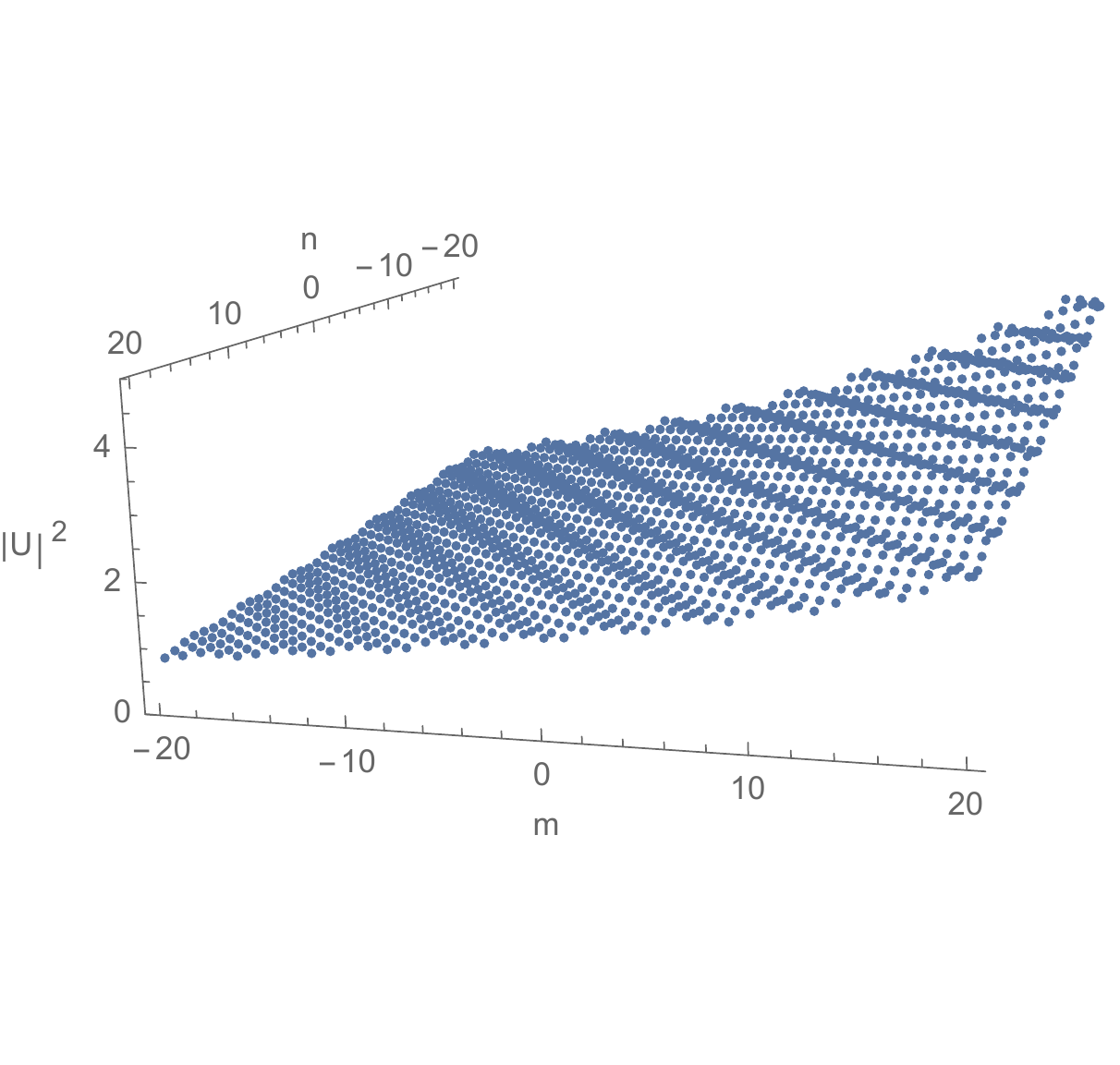}}}
\put(100,-23){\resizebox{!}{3.5cm}{\includegraphics{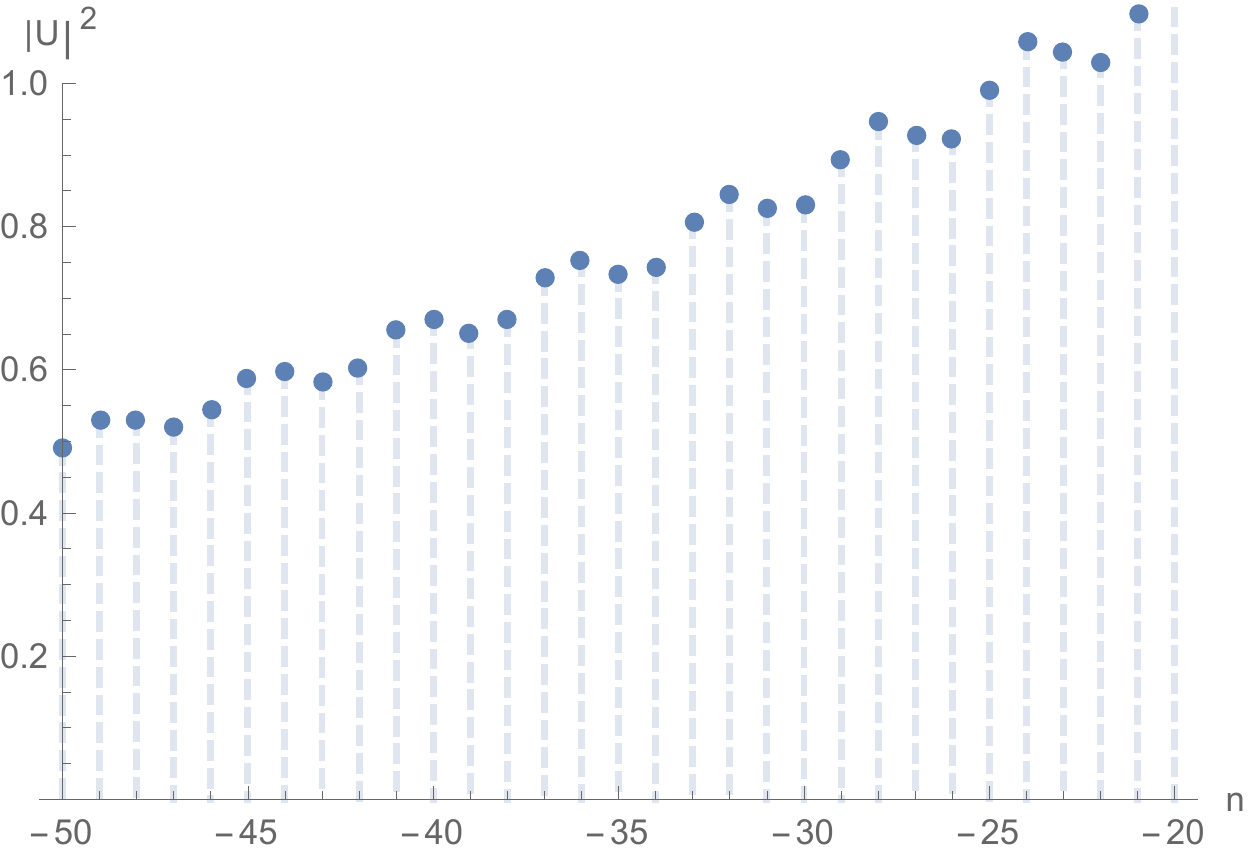}}}
\end{picture}
\end{center}
\vskip 20pt
\begin{center}
\begin{minipage}{15cm}{\footnotesize
\quad\qquad\qquad\qquad\qquad\qquad(a)\qquad\qquad\qquad\qquad\qquad\qquad\qquad\qquad\qquad\qquad\quad \qquad(b)}
\end{minipage}
\end{center}
\vskip 10pt
\begin{center}
\begin{picture}(120,80)
\put(-120,-23){\resizebox{!}{4cm}{\includegraphics{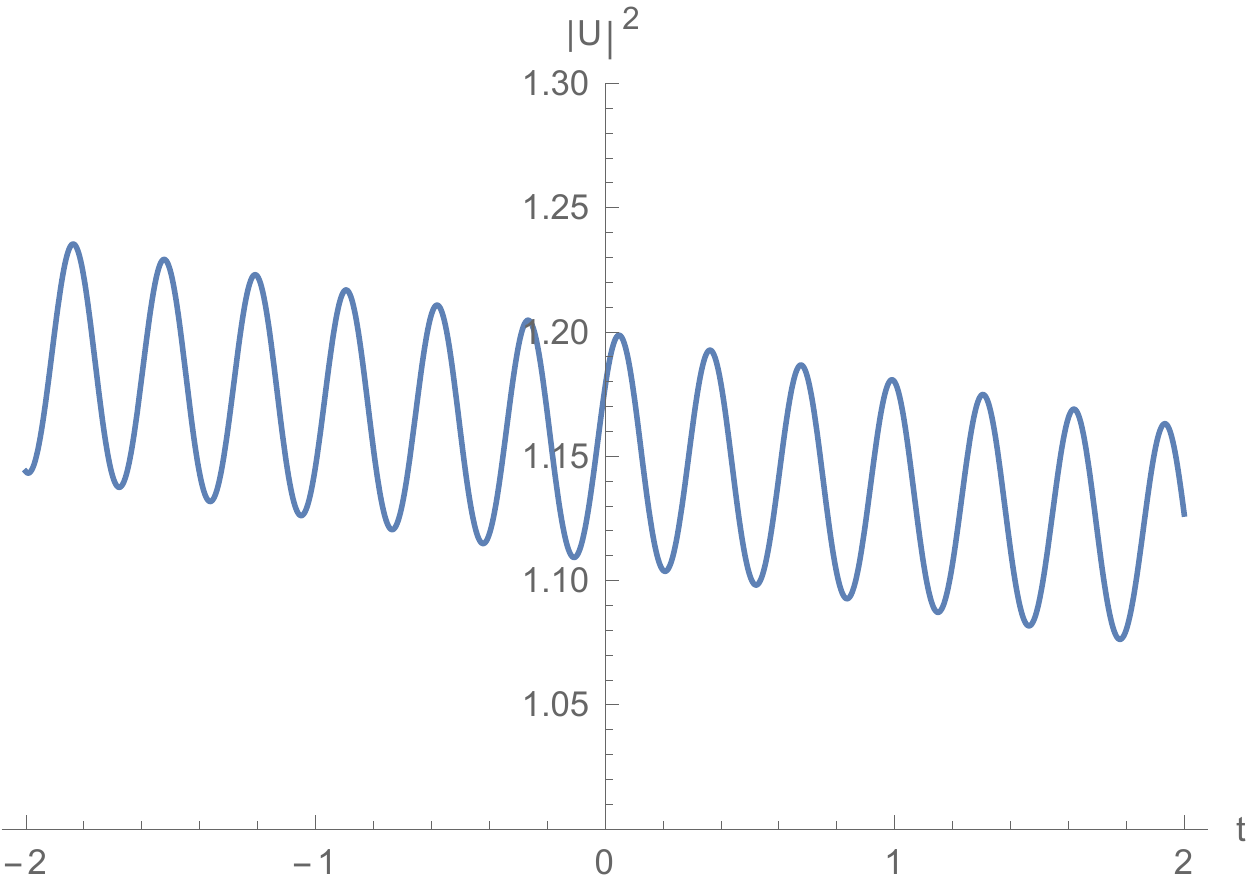}}}
\put(100,-23){\resizebox{!}{4cm}{\includegraphics{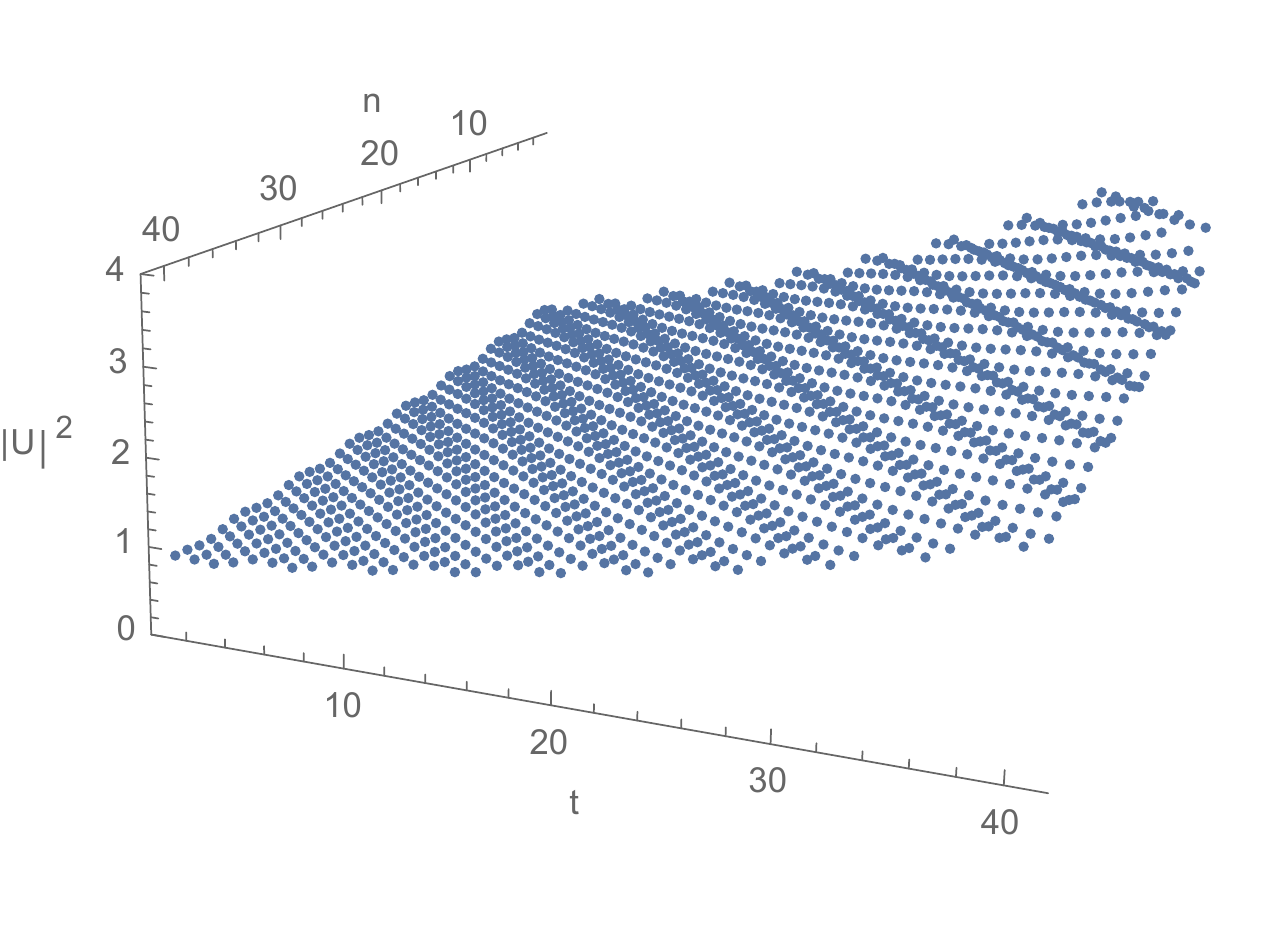}}}
\end{picture}
\end{center}
\vskip 20pt
\begin{center}
\begin{minipage}{15cm}{\footnotesize
\quad\qquad\qquad\qquad\qquad\qquad(c)\qquad\qquad\qquad\qquad\qquad\qquad\qquad\qquad\qquad\quad \qquad \quad (d)\\
{\bf Fig. 4} shape and motion with $|\mathbb{U}|^2$ given by \eqref{u-d-ca-F} for $k_1=-0.1+5i, p=2, \rho_1^0=1$ and $c_1=1+i$.
(a) 3D-plot for $\delta=1$.
(b) wave of plot (a) at $t=1$.
(c) wave of plot (a) at $n=-20$.
(d) 3D-plot for $\delta=-1$.}
\end{minipage}
\end{center}

\section{Conclusions} \label{Con}

In this paper, we have investigated local and nonlocal complex reduction for the dAKNS(-1)
equation \eqref{dnAKNS}. A local complex discrete sG equation (equation \eqref{cnd-sG} with $\sigma=1$)
and a nonlocal complex discrete sG equation (equation \eqref{cnd-sG} with $\sigma=-1$)
have been revealed. Based on the canonical DES \eqref{DES-C}, by imposing suitable constraints \eqref{cnd-sG-M1M12} on the elements $(\br_1,~\tc_1,~\bM_1)$
and $(\br_2,~\tc_2,~\bM_2)$ in the Cauchy matrix solution of the dAKNS(-1) equation \eqref{dnAKNS}, we have
obtained formal solution to the cnd-sG equation \eqref{cnd-sG}.
1-soliton solution, 2-soliton solutions and the simplest Jordan-block solution were given as three examples of the exact solutions.
Dynamics for 1-soliton solution were analyzed with graphical illustration. For the local complex discrete sG equation,
its 1-soliton solution exhibited the usual bell-type structure. For the nonlocal complex discrete sG equation,
its 1-soliton solution showed quasi-periodic phenomenon. Besides the dAKNS(-1) equation \eqref{dnAKNS}, we also discussed
local and nonlocal complex reduction for the sdAKNS(-1) equation \eqref{sdnAKNS}. Soliton solutions and Jordan-block
for the resulting local and nonlocal complex semi-discrete sine-Gordon equation are constructed.
Some discrete models of positive order AKNS-type equations, admitting Cauchy matrix solutions, have been proposed \cite{Zhao-ZNA,Zhao-JDEA}.
How to consider their local and nonlocal complex reduction and derive the solutions of the resulting
local and nonlocal complex equations are interesting questions worth consideration.

\vskip 20pt
\subsection*{Acknowledgments}
This project is supported by the National Natural Science Foundation of China (Nos. 12071432, 11401529)
and the Natural Science Foundation of Zhejiang Province (No. LY18A010033).

\vskip 20pt
\subsection*{Conflict of interest}
The authors declare that there is no conflict
of interests regarding the publication of this paper.

\end{document}